\documentclass{cccg14}
\usepackage{graphicx,amssymb,amsmath}
\usepackage{xspace}
\usepackage{refcount}
\usepackage{amssymb}
\usepackage{amsmath}
\usepackage{xcolor}

\graphicspath{{figures/}}

\newcommand{\etal}{\emph{et al.}\xspace}
\newcommand{\cyao}{\ensuremath{cY(\theta)}\xspace}
\newcommand{\cyaoOneTwenty}{\ensuremath{cY(2\pi/3)}\xspace}
\newcommand{\cyaopi}{\ensuremath{cY(\pi)}\xspace}
\newcommand{\cyaoangle}[1]{\ensuremath{cY(#1)}\xspace}
\newcommand{\creflect}{\ensuremath{c^*}\xspace}
\newcommand{\vreflect}{\ensuremath{v^*}\xspace}

\newcommand{\spanningRationCyao}{\ensuremath{6.0411}\xspace}

\title{Continuous Yao Graphs}

\author{
Luis Barba\thanks{School of Computer Science, Carleton University. Email: \texttt{jit@scs.carleton.ca, jdecaruf@cg.scs.carleton.ca, andre@cg.scs.carleton.ca, sander@cg.scs.carleton.ca}. Research supported in part by NSERC and Carleton University's President's 2010 Doctoral Fellowship.}\ \thanks{D\'epartement d'Informatique, Universit\'e Libre de Bruxelles. Email:  \texttt{lbarbafl@ulb.ac.be}}
\and
Prosenjit Bose\footnotemark[1]
\and
Jean-Lou De Carufel\footnotemark[1]
\and
Mirela Damian\thanks{Department of Computing Sciences, Villanova University. Email: \texttt{mirela.damian@villanova.edu}. Research supported by NSF grant CCF-1218814.}
\and
Rolf Fagerberg\thanks{Department of Computer Science, University of Southern Denmark. Email: \texttt{rolf@imada.sdu.dk}.}
\and
Andr\'e van Renssen\footnotemark[1]
\and
Perouz Taslakian\thanks{School of Science and Engineering, American University of Armenia. Email: \texttt{perouz.taslakian@ulb.ac.be}.}
\and
Sander Verdonschot\footnotemark[1]
}

\index{Barba, Luis}
\index{Bose, Prosenjit}
\index{De Carufel, Jean-Lou}
\index{Damian, Mirela}
\index{Fagerberg, Rolf}
\index{van Renssen, Andr\'e}
\index{Taslakian, Perouz}
\index{Verdonschot, Sander}

\begin{document}
\thispagestyle{empty}
\maketitle

\begin{abstract}
In this paper, we introduce a variation of the well-studied Yao graphs.
Given a set of points $S\subset \mathbb{R}^2$ and an angle $0 < \theta \leq 2\pi$,
we define the \emph{continuous Yao graph} $\cyao$ with vertex set $S$ and angle $\theta$ as follows. For each $p,q\in S$, we add an edge from $p$ to $q$ in $\cyao$ 
if there exists a cone with apex $p$ and aperture $\theta$ such that $q$ is the closest point to $p$ inside this cone.

We study the spanning ratio of \cyao for different values of $\theta$.
Using a new algebraic technique, we show that $\cyao$ is a spanner when $\theta \leq 2\pi /3$. We believe that this technique may be of independent interest. We also show that $\cyaopi$ is not a spanner, and that $\cyao$ may be disconnected for $\theta > \pi$.

\end{abstract}

\section{Introduction}
Let $S$ be a set of points in the plane.
The complete geometric graph with vertex set $S$ has a straight-line edge connecting each pair of points in $S$. 
Because the complete graph has quadratic size in terms of number of edges, several methods for ``approximating'' this graph with a graph of linear size have been proposed. 

A geometric $t$-spanner $H$ of $S$ is a spanning subgraph of the complete geometric graph of $S$ with the property that for all pairs of points $p$ and $q$ of $S$, the length of the shortest path between $p$ and $q$ in $H$ is at most $t$ times the Euclidean distance between $p$ and $q$.

The \emph{spanning ratio} of a spanning subgraph is the smallest $t$ for which this subgraph is a $t$-spanner. 
For a comprehensive overview of geometric spanners and their applications, we refer the reader to the book by Narasimhan and Smid \cite{NS06}.

A simple way to construct a $t$-spanner is to first partition the plane around each point $p\in S$ into a fixed number of cones\footnote{The orientation of the cones is the same for all vertices.} and then add an edge connecting $p$ to a closest vertex in each of its cones. These graphs have been independently introduced by Flinchbaugh and Jones~\cite{flinchbaugh1981strong} and Yao~\cite{yao1982constructing}, and are referred to as \emph{Yao graphs} in the literature.
%
It has been shown that Yao graphs are good approximations of the complete geometric graph
\cite{clarkson1987approximation,althofer1993sparse,bose2004approximating,bose2012piArxiv,damian2012yao,bose2012pi,el2009yao,barba2014new}.

We denote the Yao graph defined on $S$ by $Y_k$, where $k$ is the number of cones, each having aperture $\theta = 2\pi/k$.
Clarkson~\cite{clarkson1987approximation} was the first to remark that $Y_{12}$ is a $1 + \sqrt{3}$-spanner in 1987.
Alth{\"o}fer~\etal~\cite{althofer1993sparse} showed that for every $t > 1$, there is a $k$ such that $Y_k$ is a $t$-spanner. 
For $k > 8$, Bose~\etal~\cite{bose2004approximating} showed that $Y_k$ is a geometric spanner with spanning ratio at most $1 / (\cos \theta - \sin \theta)$. 
This was later strengthened to show that for $k>6$, $Y_k$ is a $1/(1-2\sin(\theta/2))$-spanner~\cite{bose2012piArxiv}. 
Damian and Raudonis~\cite{damian2012yao} proved a spanning ratio of $17.64$ for $Y_6$, which was later improved by Barba~\etal to $5.8$~\cite{barba2014new}. In~\cite{barba2014new} the authors also improve the spanning ratio of $Y_k$ for all odd values of $k\geq 5$ to $1/(1-2\sin (3\theta/8))$. In particular, they show an upper bound on the spanning ratio for $Y_5$ of $2+\sqrt{3}\approx 3.74$.
Bose~\etal~\cite{bose2012pi} showed that $Y_4$ is a $663$-spanner. For $k < 4$, El Molla~\cite{el2009yao} showed that there is no constant $t$ such that $Y_k$ is a $t$-spanner. 

Yao graphs are based on the implicit assumption that all points use identical cone orientations with respect to an extrinsic fixed direction. From a practical point of view, if these points represent wireless devices and edges represent communication links for instance, the points would need to share a global coordinate system to be able to orient their cones identically. Potential absence of global coordinate information adds a new level of difficulty by allowing each point to spin its cone wheel independently of the others. In this paper we take a first step towards reexamining Yao graphs in light of intrinsic cone orientations, by introducing a new class of graphs called \emph{continuous Yao graphs}.

Given an angle $0< \theta\leq 2\pi$, the continuous Yao graph with angle $\theta$, denoted by $\cyao$, is the graph with vertex set $S$, and an edge connecting two points $p$ and $q$ of $S$ if there exists a cone with angle $\theta$ and apex $p$ such that $q$ is the closest point to $p$ inside this cone. In contrast with the classical construction of Yao graphs, for the continuous version the orientation of the cone is arbitrary. We can imagine rotating a cone with angle $\theta$ around each point $p\in S$ and connecting it to each point that becomes the closest to $p$ inside the cone during this rotation. To avoid degenerate cases, we assume \emph{general position}, i.e., we assume that for each $p\in S$, there are no two points at the same distance from $p$.

In contrast with the Yao graph, the continuous Yao graph has the property that $\cyao \subseteq \cyaoangle{\gamma}$ for any $\theta \geq \gamma$. This property provides consistency as the angle of the cone changes and could be useful in potential applications requiring scalability. Another advantage of continuous Yao graphs over regular Yao graphs is that they are invariant under rotations of the input point set. However, unlike Yao graphs that guarantee a linear number of edges, continuous Yao graphs may have a quadratic number of edges in the worst case. (Imagine, for instance, the input points evenly distributed on two line segments that meet at an angle $\alpha < \pi$. For any $\theta < \alpha$, \cyao includes edges connecting each point on one line segment to each point on the other line segment.)

In this paper, we focus on the spanning ratio of the continuos Yao graph.
In Section~\ref{section:Small cones}, we show that $\cyao$ has spanning ratio at most $1/(1-2\sin(\theta/4))$ when $\theta < 2\pi/3$. 
However, the argument used in this section breaks when $\theta = 2\pi/3$. To deal with this case, we introduce a new algebraic technique based on the description of the regions where induction can be applied.
To the best of our knowledge, this is the first time that algebraic techniques are used to bound the spanning ratio of a graph. 
As such, our technique may be of independent interest. In Section~\ref{section:The 2pi/3 cyao}, we use this technique to show that $\cyaoangle{2\pi/3}$ is a $\spanningRationCyao$-spanner.
In Section~\ref{Section:Other angles}, we study the case when $\theta>2\pi/3$. Using elliptical constructions, we are able to show that $\cyaoangle{\pi}$ is not a constant spanner. While the algebraic techniques presented in Section~\ref{section:The 2pi/3 cyao} appear to extend beyond $2\pi/3$, it remains open whether or not $\cyao$ with angle $2\pi/3 <\theta < \pi$ is a constant spanner.
Finally, we study the connectivity of $\cyao$ and show that $\cyao$ is connected provided that $\theta \leq \pi$. Moreover, for $\theta > \pi$, there exist point sets for which $\cyao$ is not connected.

\vspace{-0.5em}
\section{Continuous Yao for narrow cones}\label{section:Small cones}

In this section, we study the spanning ratio of \cyao for $\theta < 2\pi/3$. 
In this case, we make use of an inductive proof similar to those used to bound the spanning ratio of Yao graphs~\cite{barba2014new}.

\vspace{-0.2em}
\begin{lemma}
 \label{lem:basicyao}\emph{[Lemma 1 of~\cite{barba2014new}]}
Let $a$, $b$ and $c$ be three points such that $|ac| \leq |ab|$ and $\angle bac \leq \alpha < \pi$. Then\vspace{-.1in}
$$ |bc| \leq |ab| - \left( 1 - 2 \sin (\alpha/2) \right)  |ac|~. $$
\end{lemma}
Given two points $a$ and $b$ of \cyao,
let $C_{ab}$ be the cone with apex $a$ and $b$ on its angle bisector.
The cone $C_{ba}$ is defined analogously.

\begin{theorem}
 \label{thm:contyaosmall}
 The graph \cyao has spanning ratio at most $1 / (1 - 2 \sin (\theta/4))$ for $0 < \theta < 2\pi/3$.
\end{theorem}
\begin{proof}
 We need to show that there exists a path of length at most $1 / (1 - 2 \sin (\theta/4))  |ab|$ between any two vertices $a$ and $b$. We prove this by induction on the distance $|ab|$. In the base case $a$ and $b$ form the closest pair. Hence, the edge $ab$ is added by any cone of~$a$ that contains $b$, as no other vertex can be closer to~$a$.
 
 For the inductive step, we assume that the theorem holds for any two vertices whose distance is less than $|ab|$. If the edge $ab$ is in the graph, the proof is finished, so assume that this is not the case. That means that there is a vertex closer to $a$ in every cone with apex $a$ that contains $b$. In particular, this also holds for the cone $C_{ab}$. Let $n_a$ be the vertex that is closest to $a$ in $C_{ab}$. Since $C_{ab}$  has aperture $\theta$, the angle $\angle n_aab$ is at most $\theta / 2$, and Lemma~\ref{lem:basicyao} gives us that 
 $|bn_a| \leq |ab| - (1 - 2 \sin (\theta/4))  |an_a|$. Note that since $\theta < 2\pi/3$, we have that $\theta/4 < \pi/6$, which means that $1 - 2 \sin (\theta/4) > 0$ and hence $|bn_a| < |ab|$. Therefore our inductive hypothesis applies to $n_a$ and $b$, which tells us that there exists a path between them of length at most $1 / (1 - 2 \sin (\theta/4)) |bn_a|$. Adding the edge $an_a$ to this path yields a path between $a$ and $b$ of length at most

\vspace{-1em}
\begin{align*}
& \, |an_a| + \frac{1}{1 - 2 \sin (\theta/4)} |bn_a| \leq\\
& \, |an_a| + \frac{1}{1 - 2 \sin (\theta/4)}\left( |ab| - (1 - 2 \sin (\theta/4))  |an_a| \right)=\\
& \, |an_a| + \frac{1}{1 - 2 \sin (\theta/4)} |ab| - |an_a|
= \, \frac{1}{1 - 2 \sin (\theta/4)} |ab|.
\end{align*}
This completes the proof.
\end{proof}

\section{The graph $\boldsymbol{\cyaoOneTwenty}$ is a Spanner}\label{section:The 2pi/3 cyao}

Let $t\approx \spanningRationCyao$ be the largest root of the polynomial
$p(t) = -25 + 90 t - 39 t^2 - 246 t^3 + 363 t^4 + 138 t^5 - 589 t^6 + 
 216 t^7 + 291 t^8 - 204 t^9 - 84 t^{10} + 6 t^{11} + 2 t^{12}.$
In this section, we prove that \cyaoOneTwenty is a $t$-spanner.
That is, we show that for any two points $a$ and $b$ in \cyaoOneTwenty, there exists a path from $a$ to $b$ of length at most $t\,|ab|$. The way we derive this polynomial will become clear by the end of this section.

The proof proceeds by induction on the rank of the distance $|ab|$ among all distances between vertices of \cyaoOneTwenty.
In the base case,
$a$ and $b$ define the closest pair among the points of \cyaoOneTwenty.
Hence, the edge $ab$ is added by any cone of $a$ that contains $b$,
as no other vertex can be closer to $a$.

We spend the remainder of this section proving the inductive step.
Assume that the result holds for any two points whose distance is smaller than $|ab|$.
Without loss of generality,
assume that $a=(0,0)$ and $b=(1,0)$,
so that $|ab| = 1$.
We start with a simple observation that follows from the general position assumption.
Define $I_{ab} = \{p\in \mathbb{R}^2 : |ap| + t|pb| \leq t|ab|\}$ be the \emph{inductive set of $a$ with respect to $b$}
(see Fig.~\ref{fig:Inductive Region}).
\begin{figure}[t]
\centering
\includegraphics[width=0.45\linewidth]{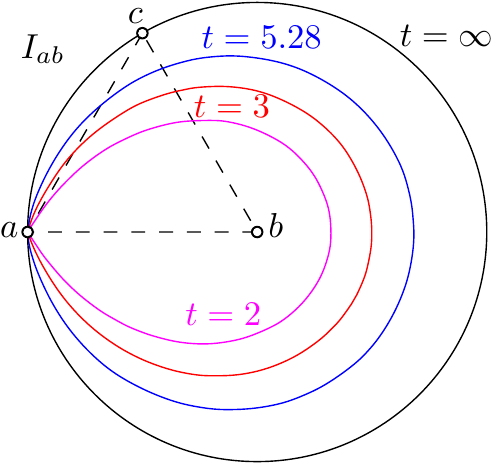}
\caption{\small The inductive set $I_{ab}$ for different values of $t$.}
\label{fig:Inductive Region}
\end{figure}
Symmetrically,
let $I_{ba} = \{p\in \mathbb{R}^2 : |bp| + t|pa| \leq t|ba|\}$ be the inductive set of $b$ with respect to $a$.

\begin{lemma}\label{lemma:Contained in circle}
The inductive set $I_{ab}$ is contained in the disk $D$ with center $b$ and radius $|ab|$.
Moreover, any point $p\neq a$ on the boundary of $D$ lies outside of $I_{ab}$.
\end{lemma}
\vspace{-0.8em}
\begin{proof}
Let $p \neq a$ be a point in $I_{ab}$.
Because $|ap|>0$ and $t > 1$, we have that 
$t  |pb| < |ap| + t|pb| \leq t|ab|$.
Consequently,
$p$ lies strictly inside the circle with center $b$ and radius $|ab|$.
\end{proof}

\vspace{-0.5em}
Recall that $C_{ab}$ denotes the cone with apex $a$ and $b$ on its angle bisector.
Let $n_a$ and $n_b$ be the neighbors of $a$ and $b$ in cones $C_{ab}$ and $C_{ba}$, respectively.
The inductive set $I_{ab}$ satisfies the \emph{inductive property}:
if $n_a \in I_{ab}$,
then there is a path from $a$ to $b$ with length at most $t|ab|$.
Indeed,
because $n_a \in I_{ab}$, Lemma~\ref{lemma:Contained in circle} implies that $|n_a b| < |ab|$.
Therefore,
we can apply the induction hypothesis and obtain a path from $n_a$ to $b$ of length at most $t|n_a b|$. 
Because $n_a\in I_{ab}$, adding the edge $an_a$ to this path yields a path from $a$ to~$b$ 
of length at most $|an_a| + t|n_a b| \leq t|ab|$ as desired.
The inductive set $I_{ba}$ has an analogous inductive property.

Note that if $n_a \in I_{ab}$
or $n_b \in I_{ba}$,
then we are done by the inductive property.
Thus, we assume that $n_a \not\in I_{ab}$ and $n_b \not\in I_{ba}$.
Since $a = (0,0)$ and $b = (1,0)$,
the set of points on the boundary of $I_{ab}$ satisfy
\begin{align}
\nonumber
 & ((-2 + x) x + y^2)^2 \, t^4+ \,   (x^2 + y^2)^2&\\ 
\label{eq:quartic}
- \, & 2 (2 + (-2 + x) x + y^2) (x^2 + y^2) \, t^2
 = 0,
\end{align}
which defines a quartic curve in $x$ and $y$. 
Let $c$ and $c^*$ be the intersection points of the boundaries of $C_{ab}$ and $C_{ba}$ and assume that $c$ lies above $c^*$;
see Fig.~\ref{fig:Configuration}.
\begin{figure}[t]
\centering
\includegraphics[width=0.6\linewidth]{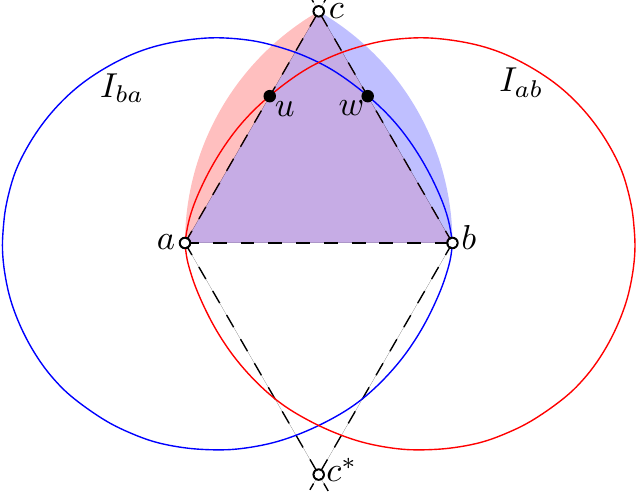}
\caption{\small The inductive sets $I_{ab}$ and $I_{ba}$ are shown. The circular sectors where $n_a$ and $n_b$ can lie are depicted in light blue and light red, respectively.}
\label{fig:Configuration}
\end{figure}
Because the triangles $\triangle abc$ and $\triangle abc^*$ are equilateral,
we have $c = (1/2, \sqrt{3}/2)$ and $c^* = (1/2, -\sqrt{3}/2)$.
Let
\begin{equation}
\label{eq:Coordinates of u}
u = \left ( \frac{t(t-2)}{2(t^2-1)}, \frac{\sqrt{3} \, t(t-2)}{2(t^2-1)}   \right ) \approx (0.3438, 0.5956)
\end{equation}
be the intersection point of the boundary of $I_{ab}$ with the segment $ac$. 
Symmetrically,
let
$$w =\left ( 1-\frac{t(t-2)}{2(t^2-1)}, \frac{\sqrt{3} \, t(t-2)}{2(t^2-1)}   \right ) \approx (0.6561, 0.5956)$$
be the intersection of the boundary of $I_{ba}$ with the segment $bc$.
There are two cases to deal with. Either $(i)$ $n_a$ and $n_b$ lie on the same side of the $x$-axis
or $(ii)$ they lie on opposite sides.

Given three points $x$, $y$ and $y'$ such that $|xy| = |xy'|$,
we denote by $\mathcal{C}(x,y,y')$ the circular sector with apex $x$ that is contained between $xy$ and $xy'$,
counter-clockwise.\vspace{.05in}

\textbf{Case $\boldsymbol{(i)}$} 
Assume first that $n_a$ and $n_b$ both lie above the $x$-axis.
Because $n_a$ and $n_b$ lie in the circular sectors $\mathcal C(a,b,c)$ and $\mathcal C(b, c, a)$, respectively, we have that $|n_a n_b| < |ab|$.
Therefore, we can apply induction on $n_an_b$ to obtain a path $\varphi_{n_an_b}$ from $n_a$ to $n_b$ of length at most $t|n_a n_b|$.
Consider the path $\varphi_{ab} = an_a\cup \varphi_{n_an_b}\cup n_b b$ from $a$ to $b$. 
We show that the length of $\varphi_{ab}$ is at most $t|ab| = t$.
To this end, we provide a bound on the length of the segment $n_an_b$.

\vspace{-0.5em}
\begin{lemma}\label{lemma:Case one maximized length}
In the configuration of Case $(i)$ depicted in Fig.~\ref{fig:Configuration},
$|n_an_b| \leq |uc| = |wc| = |uw|$.
\end{lemma}
\vspace{-1em}
\begin{proof}
Recall that $n_a$ must lie in the circular sector $\mathcal C(a,b,c)$. Moreover,
because we assumed that $n_a$ lies outside of $I_{ab}$, $n_a$ lies in the region $\mathcal C(a, b,c) \setminus I_{ab}$. 
Let $N_a$ be the convex hull of $\mathcal C(a, b,c) \setminus I_{ab}$ and let $v$ be the intersection point between $I_{ab}$ and the circular arc of $\mathcal C(a, b, c)$; see Fig.~\ref{fig:Neighbor regions Case 1}. Analogously, let $v'$ be the intersection between $I_{ba}$ and the circular arc of $\mathcal C(b,c,a)$.
Then, $N_a$ is bounded by the segments $uc$, $uv$ and the circular arc joining $v$ and $c$ with center $a$ and radius 1. We define $N_b$ analogously as the convex hull of $\mathcal C(b, c, a)\setminus I_{ba}$.

\begin{figure}[htb]
\centering
\includegraphics[width=0.8\linewidth]{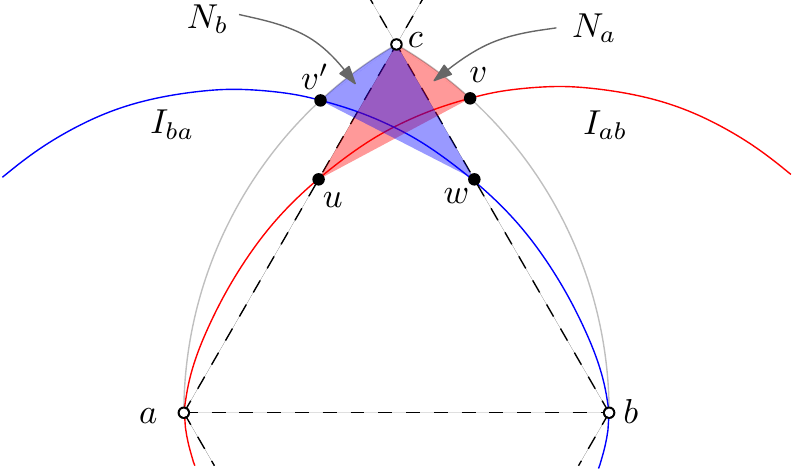}
\caption{\small The neighbor regions of $a$ and $b$ in Case $(i)$.}
\label{fig:Neighbor regions Case 1}
\vspace{-.1in}
\end{figure}

Because $n_a\in N_a$ and $n_b\in N_b$, we get an upper bound on the distance between $n_a$ and $n_b$ by computing the maximum distance between a point in $N_a$ and a point in $N_b$. We refer to two points realizing this distance as a  \emph{maximum $N_a$-$N_b$-pair}.
Since the Euclidean distance function is convex and since both $N_a$ and $N_b$ are convex sets, a maximum $N_a$-$N_b$-pair must have one point on the boundary of $N_a$ and another on the boundary of $N_b$.

In fact, we claim that we need only to consider the boundaries of the triangles $\triangle(u,v,c)\subset N_a$ and $\triangle(w, c, v')\subset N_b$ to find a maximum $N_a$-$N_b$-pair.
To prove this claim, consider the lune defined by $N_a\setminus \triangle(u,v,c)$. For any point $x$ in this lune, consider its farthest point $f(x)$ in $N_b$ and notice that the circle with center on $f(x)$ that passes through $x$ leaves either $c$ or $v$ outside (or both). This is because the radius of this circle is smaller than the radius of the circular arc on the boundary of $N_a$; see Fig.~\ref{fig:Neighbor regions Case 1}. Therefore, either $c$ or $v$ is farther than $x$ from $f(x)$ and hence, the maximum $N_a$-$N_b$-pair cannot have an endpoint in this lune. That is, the maximum $N_a$-$N_b$-pair includes a point on the boundary of the triangle $\triangle(u,v,c)$. The same argument holds for $\triangle(w, c, v')$ and $N_b$ proving our claim.

As we know the coordinates of the boundary vertices of $\triangle(u,v,c)$ and $\triangle(w, c, v')$, we can verify that 
$(u,c)$, $(c, w)$ and $(u, w)$ are all 
maximum $N_a$-$N_b$-pairs 
(notice that this is true for any $t > 1$).
\end{proof}

Because the length of $n_an_b$ is at most $|uc|$, and since $|an_a|$ and $|bn_b|$ are both at most 1,
the length of the path $\varphi_{ab} = an_a\cup \varphi_{n_an_b}\cup n_b b$ is at most $2+ t|uc|$
by Lemma~\ref{lemma:Case one maximized length}. 
We now prove that $2+ t|uc|\leq t|ab|$.
Since $a=(0,0)$, $b=(1,0)$, $c=(1/2,\sqrt{3}/2)$
and $|au| = \mu = \frac{t(t-2)}{t^2-1}$,
 the inequality $2+t|uc| \leq t|ab|$ is equivalent to 
  \vspace{-.1in}
 \begin{equation*} 
2+ t\left(1-\frac{t(t-2)}{t^2-1}\right) \leq t
 \vspace{-.1in}
\end{equation*}
which is true,
provided that
$t^3 -4t^2+2 \geq 0$ and $t>1$.
Since $t = \spanningRationCyao$ is bigger than the largest real root of $x^3 -4x^2+2$,
we are done.
Therefore,
whenever we are in the configuration of Case $(i)$,
we can apply induction and obtain a path $\varphi_{ab}$ from $a$ to $b$ of length at most $2 + t|uc| \leq t|ab|$.\vspace{.05in}

\textbf{Case $\boldsymbol{(ii)}$} The proof of Case $(ii)$ is a bit more involved but follows the same line of reasoning as the proof of Case $(i)$.
If $n_a$ and $n_b$ lie on different sides of $ab$, we can assume without loss of generality that~$n_a$ lies below the $x$-axis while $n_b$ lies above it.
Recall that $\creflect$ is the intersection of the boundaries of $C_{ab}$ and $C_{ba}$ that lies below the $x$-axis.

\begin{figure*}[htb]
\centering
\includegraphics[width=0.82\textwidth]{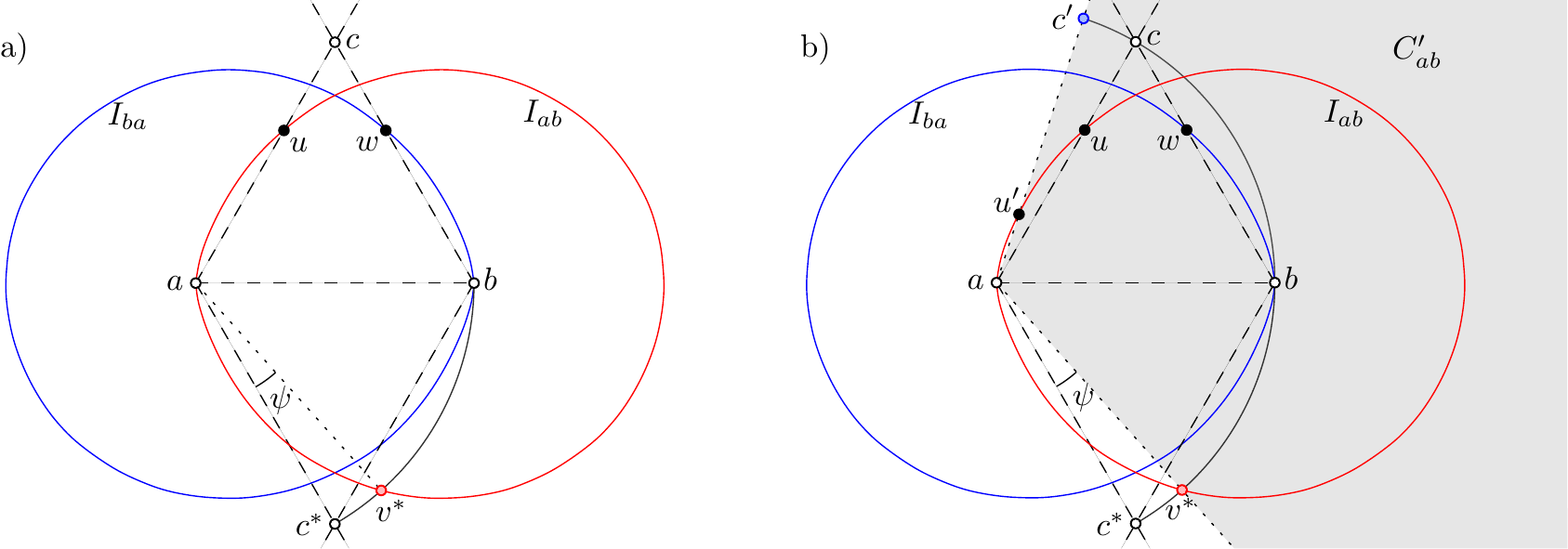}
\vspace{-0.5em}
\caption{\small a) Point $\vreflect$ and angle $\psi = \angle \vreflect a \creflect$ b) Cone $C'_{ab}$ is obtained by rotating $C_{ab}$ counter-clockwise $\psi$ degrees.}
\vspace{-0.5em}
\label{fig:Rotated Configuration}
\end{figure*}

Since $ab$ is not an edge of \cyaoOneTwenty,
$n_a$ must lie inside $\mathcal{C}(a,\creflect,b)$.
Let $\vreflect$ be the intersection of the boundary of $I_{ab}$ with the circular arc of $\mathcal{C}(a,\creflect,b)$; see Fig.~\ref{fig:Rotated Configuration}.
This intersection point always exists because
$b$ lies inside $I_{ab}$
and $\creflect$ lies outside of $I_{ab}$ by Lemma~\ref{lemma:Contained in circle}.
The circular arc of $\mathcal{C}(a,\creflect,b)$ is part of the circle defined by $x^2+y^2 = 1$.
Therefore,
from~\eqref{eq:quartic},
 \vspace{-.07in}
\begin{align}
\label{expression vreflect}
\vreflect =& \, \left(\frac{t^2+2t-1}{2t^2},-\frac{t-1}{2t^2}\sqrt{(t+1)(3t-1)}\right) \\
\nonumber
\approx& \, (0.6518, -0.7583) \enspace.
\vspace{-.1in}
\end{align}

Let $\psi = \angle \vreflect a \creflect$; see Fig.~\ref{fig:Rotated Configuration}a.
Since $\psi = \pi/3 - \angle ba\vreflect$,
from~\eqref{expression vreflect}
we have $tan(\psi) $\vspace{-.09in}
\begin{align}
\nonumber
=&\,  \tan(\pi/3 - \angle ba\vreflect)
\nonumber
= \,  \frac{\tan(\pi/3) - \tan(\angle ba\vreflect)}{1+\tan(\pi/3)\,\tan(\angle ba\vreflect)}\\ 
\label{expression tan psi}
= \, & \frac{\sqrt{3}\left(t^2+2t-1\right)-(t-1)\sqrt{(t+1)(3t-1)}}
{t^2+2t-1+\sqrt{3}(t-1)\sqrt{(t+1)(3t-1)}}
\end{align}
from which $\tan(\psi) \approx 0.1885$ and hence, $\psi \approx 10.6800^\circ$.
Consider the cone $C'_{ab}$
(respectively the point $c'$)
obtained by rotating $C_{ab}$
(respectively $c$)
counter-clockwise around $a$ by an angle $\psi$.
Note that $\mathcal{C}(a,\vreflect,b) \subset I_{ab}$; see Fig.~\ref{fig:Rotated Configuration}b.
Let $n_a'$ be the neighbor of $a$ inside $C'_{ab}$. 
If $n_a'$ lies inside $I_{ab}$, we are done by the inductive property.
Therefore, assume that $n_a' \not\in I_{ab}$.
Because $\mathcal{C}(a,\vreflect,b) \subset I_{ab}$, $n_a'$ cannot lie inside $\mathcal{C}(a,\vreflect,b)$ and hence,
$n_a'$ must lie above the $x$-axis. 
Let $N'_a$ be the convex hull of 
$\mathcal C(a, c', b) \setminus I_{ab}$. Then $n'_a$ 
must lie inside of $N'_a$; see Fig.~\ref{fig:Zoom Case 2} for an illustration.
As in Case $(i)$, $n_b$ must lie inside of the region $N_b$ being the convex hull of $\mathcal C(b, c, a) \setminus I_{ba}$.

Let $u' \in ac'$ be the intersection of the boundaries of $C'_{ab}$ and $I_{ab}$
(see Fig.~\ref{fig:Zoom Case 2}).
From~\eqref{expression tan psi},
the equation of the line supported by $a$ and $c'$ is
\begin{align*}
y &= \tan(\pi/3 + \psi)\,x = \frac{\tan(\pi/3) + \tan(\psi)}{1-\tan(\pi/3)\,\tan(\psi)} \,x\\
&= \frac{\sqrt{3}\left(t^2+2t-1\right)+(t-1)\sqrt{(t+1)(3t-1)}}
{-\left(t^2+2t-1\right)+\sqrt{3}(t-1)\sqrt{(t+1)(3t-1)}} \,x \enspace.
\end{align*}
Thus,
the $x$-coordinate of $u'$ is given by the expression
\begin{align*}
&\frac{1}{4t^2(t^2-1)}\Big(5t^4-2t^3+2t^2+2t-1\\
&\phantom{\frac{1}{4t^2(t^2-1)}\Big(}-\sqrt{3}(t-1)(t^2+4t-1)\sqrt{(t+1)(3t-1)}\Big)
\end{align*}
and the $x$-coordinate of $c'$ is given by the expression
$$\frac{-(t^2+2t-1)+\sqrt{3}(t-1)\sqrt{(t+1)(3t-1)}}{4t^2} \enspace.$$
Thus,
$u' \approx (0.1124, 0.3207)$
and $c'\approx(0.3308,0.9436)$.

A proof similar to that of Lemma~\ref{lemma:Case one maximized length} (moved to the appendix due to space constraints) yields the following result. 
\begin{lemma}\label{lemma:Maximum length Case 2}
In the configuration of Case $(ii)$, the distance between $n_a'$ and $n_b$ is at most $|u'c|$.
\end{lemma}

By Lemma~\ref{lemma:Maximum length Case 2}, the distance between $n_a'$ and $n_b$ is at most $|u'c|< 1$. Therefore, we can apply the induction hypothesis to obtain a path $\varphi_{n_a'n_b}$ from $n_a'$ to $n_b$ of length at most $t|n_a'n_b|$.

Let $\varphi_{ab} = a n_a' \cup \varphi_{n_a'n_b} \cup n_b b$ be a path from $a$ to $b$. 
Similarly to what we observed in Case $(i)$,
the length of $\varphi_{ab}$ is at most $2 + \varphi_{n_a'n_b} \leq 2 + t|u'c|$ by Lemma~\ref{lemma:Maximum length Case 2}.

We now prove that $2 + t|u'c| \leq t|ab|$.
Since $a=(0,0)$, $b=(1,0)$ and $c=(1/2,\sqrt{3}/2)$, using the exact expressions for $u'$
we find that $2 + t|u'c| \leq t|ab|$,
provided that $p(t)  = -25 + 90 t - 39 t^2 - 246 t^3 + 363 t^4 + 138 t^5 - 589 t^6 + 216 t^7 + 291 t^8 - 204 t^9 - 84 t^{10} + 6 t^{11} + 2 t^{12} \geq 0$.
Because we chose  $t\approx \spanningRationCyao$ to be equal to the largest real root of $p$, we infer that $2 + t|u'c| \leq t|ab|$.
Therefore,
whenever we are in the configuration of Case $(ii)$,
we can apply induction and obtain a path $\varphi_{ab}$ from $a$ to $b$ of length at most $2 + t|u'c| \leq t|ab|$.

\begin{figure}[b]
\centering
\includegraphics[width=0.8\linewidth]{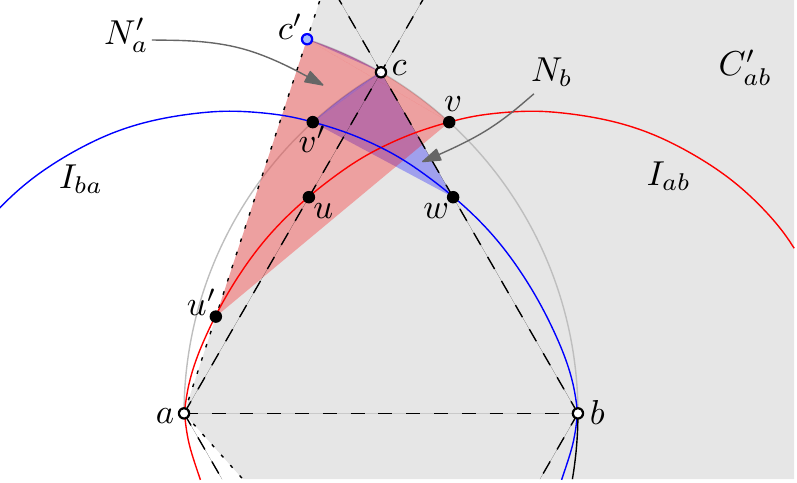}
\caption{\small $N'_a$, $N_b$ and maximum $N'_a$-$N_b$-pair $(u', c)$.}
\label{fig:Zoom Case 2}
\end{figure}

In summary,
given any two points $a$ and $b$ of \cyaoOneTwenty and a constant $t\approx \spanningRationCyao$,
we can construct a path from $a$ to $b$ which uses edges of \cyaoOneTwenty and has length at most $t|ab|$. We obtain the following result.

\vspace{-0.3em}
\begin{theorem}
 \label{thm:contyao}
 The graph  \cyao has spanning ratio at most $\spanningRationCyao$ if $\theta = 2\pi/3$, or $\min\left\{\spanningRationCyao, \frac{1}{1 - 2 \sin (\theta/4)}\right\}$ if $\theta < 2\pi/3$.
\end{theorem}

\vspace{-1.5em}
\section{Larger angles}\label{Section:Other angles}

\vspace{-0.5em}
Theorem~\ref{thm:contyao} provides upper bounds for the spanning ratio of $\cyao$ for values of $\theta\leq 2\pi/3$.
But what happens when $\theta$ is larger than $2\pi/3$?
The next result shows that if $\theta$ is very large, the graph can be disconnected.

\begin{figure}[ht]
 \centering
 \includegraphics[width=0.7\linewidth]{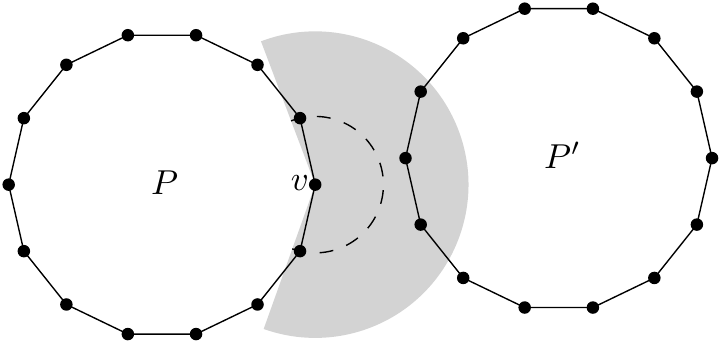}
 \caption{\small  \cyao\ can be disconnected when $\theta > \pi$.}
 \label{fig:Disconnected}
\end{figure}

\vspace{-0.5em}
\begin{theorem}
 For $\theta > \pi$, there are point sets for which \cyao is disconnected.
\end{theorem}
\begin{proof}
 Let $\theta = \pi + \varepsilon$, for any $\varepsilon > 0$. Take a regular polygon $P$ with interior angles of at least $\pi - \varepsilon/2$ radians, and let $P'$ be a copy of $P$. Now place $P$ and $P'$ such that the distance between them is larger than the distance between two consecutive vertices on $P$ (see Fig.~\ref{fig:Disconnected}). Consider a vertex $v$ on $P$. The exterior angle at $v$ is at most $2\pi - (\pi - \varepsilon/2) = \pi + \varepsilon/2$ radians. As this is less than $\theta$, any cone with apex $v$ will include one of $v$'s neighbors on $P$. And since the distance between $P$ and $P'$ is larger than the distance between $v$ and its neighbors, $v$ will never connect to a vertex on $P'$. As the choice of $v$ was completely arbitrary, and $P'$ is a duplicate of $P$, this implies that no edge of \cyao will connect $P$ to $P'$.
\end{proof}

Indeed, $\pi$ is the true breaking point here: the continuous Yao graph with $\theta \leq \pi$ is always connected (for a proof, see Appendix~\ref{Appendix}).
Next we show that, despite being connected, $\cyaopi$ is not a constant spanner.

\begin{theorem}
 \label{thm:CyaoPi No Spanner}
 The continuous Yao graph \cyaopi is not a constant spanner.
\end{theorem}
\begin{proof}
 Consider two points $p$ and $q$ at unit distance. We will add points such that the shortest path between $p$ and $q$ in \cyaopi is arbitrarily long. The construction is illustrated in Fig.~\ref{fig:LowerBound}. We place these additional points on an ellipsis that is obtained from the circle with diameter $pq$ by stretching it vertically by a factor of $2r$, for a fixed real $r \ge 1$.  (Fig.~\ref{fig:LowerBound}a). We start by placing four points, each at distance $1/2$ from $p$ or $q$ (Fig.~\ref{fig:LowerBound}b). Then we place points at distance $1/2$ from these points, and so on, until the two chains meet (when the distance between the last point on the upwards chain from $p$ and the symmetric point from $q$ is less than $1/2$: Fig.~\ref{fig:LowerBound}c).
 
 With these points, any half-plane through a vertex $v$ 
that contains vertices on the other side of the ellipsis also contains a neighbor of $v$.
As these neighbors are always closer (before the end of the chain), no diagonals are created. Thus \cyaopi forms a convex polygon, following the contour of the ellipsis (Fig.~\ref{fig:LowerBound}d).

\begin{figure}[t]
 \centering
 \includegraphics[width = 0.85\linewidth]{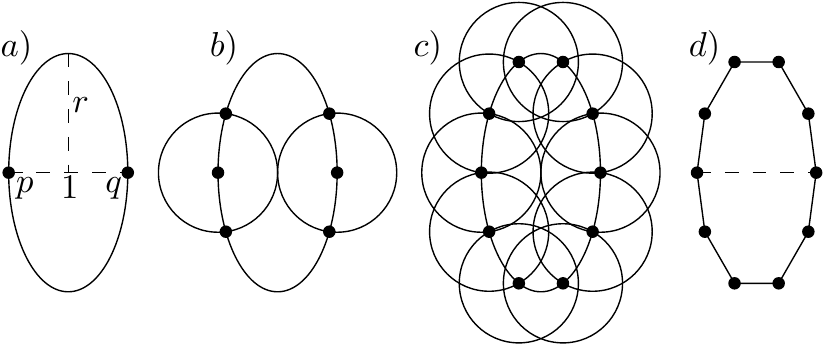}
 \caption{\small  Establishing a lower bound for the spanning ratio of \cyao\ for large values of $\theta$.}
 \label{fig:LowerBound}
\end{figure}

 As we increase $r$, the number of vertices on each chain grows. When the chains each have $k$ vertices, the shortest path between $p$ and $q$ has length at least $2  k/2 = k$. Since the distance between $p$ and $q$ remains fixed, and we can make $r$ arbitrarily large, there is no constant $t$ such that \cyaopi is a $t$-spanner.
\end{proof}

\section{Conclusions}
We introduced a new class of graphs, called continuous Yao graphs, and studied their spanning properties. We showed that, for any angle $0 < \theta \le 2\pi/3$, the continuous Yao graph \cyao is a spanner, whereas for $\pi \leq \theta \leq 2 \pi$, it is not. Furthermore, we showed that \cyao is connected for $0 < \theta \le \pi$, and possibly disconnected for $\theta > \pi$. The question whether \cyao is a spanner for $2\pi/3 < \theta < \pi$ remains open. While the construction in the proof of Theorem~\ref{thm:CyaoPi No Spanner} does give a lower bound on the spanning ratio of the continuous Yao graphs in this range, this bound seems hard to express in terms of $\theta$. For the upper bound, the proof from Section~\ref{section:The 2pi/3 cyao} appears to extend beyond $2\pi/3$, but we have not yet determined where the breaking point lies.
%


An alternative 
problem variant that maintains a linear number of edges in the output graph is one that permits each point to randomly select an  initial orientation of the entire cone wheel (as opposed to sweeping one cone continuously around the apex point). 
From Theorem~\ref{thm:CyaoPi No Spanner} we obtain as a corollary that there are point sets for which the Yao graph $Y_2$ is not a spanner, regardless of the orientation of the cones.
However, Theorem~\ref{thm:contyao} leaves open the possibility that $Y_3$ and above \emph{are} spanners under these conditions.



\section*{Acknowledgement}

The research for this paper was initiated at the first \emph{Workshop on Geometry and Graphs}, organized at the Bellairs Research Institute, March 10-15, 2013.


\small
\bibliographystyle{unsrt}
\bibliography{contyao}

\newpage
\appendix

\section{Omitted proofs}\label{Appendix}

\textbf{Lemma~\ref{lemma:Maximum length Case 2}}\emph{
In the configuration of Case $(ii)$, the distance between $n_a'$ and $n_b$ is at most $|u'c|$.
}
\begin{proof}
Because $n'_a\in N'_a$ and $n_b\in N_b$, we obtain an upper bound on the distance between $n'_a$ and $n_b$ by computing the maximum distance between a point in $N'_a$ and a point in $N_b$.
Using the same arguments as in the proof of Lemma~\ref{lemma:Case one maximized length}, we can show that the maximum distance is achieved by a point on the boundary of $N'_a$ and a point on the boundary of $N_b$.
We refer to a pair of points that realizes this maximum distance as a \emph{maximum $N'_a$-$N_b$-pair}.

One can verify that every point in $N_b$ is farther from $u'$ than from any other point in $N'_a$. Therefore, it suffices to find the point farthest from $u'$ in $N_b$. Note also that the circle centered at $u'$ that passes through any point in the circular arc of $N_b$ does not contain $c$. Therefore, it suffices to find the point farther from $u'$ in the boundary of the triangle $\triangle(w, c, v')\subset N_b$.

As we have exact expressions for $u'$ and for the vertices on the boundary of $\triangle(w, c, v')$, we can verify that the maximum $N'_a$-$N_b$-pair is found when when $n'_a = u'$ and $n_b = c$, proving our result.
\end{proof}

\begin{theorem}\label{theorem: cyan pi is connected}
 For $\theta \leq \pi$, the continuous Yao graph \cyao is connected.
\end{theorem}
\begin{proof}
 Consider a set $C_r$ of cones whose union is exactly the right half-plane. Such a set can be constructed by starting with the cone whose left boundary aligns with the positive $y$-axis, and rotating by $\pi - \theta$ degrees until the right boundary aligns with the negative $y$-axis. Since $\theta \leq \pi$, this set is non-empty. 
 Now, if a vertex $v$ is not a rightmost vertex, there is a cone $C$ in $C_r$ that is not empty. Since $C$ is completely contained in the right half-plane, the closest vertex in $C$ must lie further to the right than $v$. Thus, there is an edge connecting $v$ to a vertex to its right. Since we only have finitely many points, by repeating this, we obtain a path from any vertex to a rightmost vertex. Finally, by slightly rotating the right half plane at each rightmost point (so that it includes only rightmost vertices), we obtain a path connecting all rightmost vertices (if several rightmost vertices exist). 
 Thus, by concatenating the paths from two arbitrary points $a$ and $b$ to rightmost vertices to the path connecting these rightmost vertices, we obtain a path between $a$ and $b$.
\end{proof}

\end{document}